\documentclass[conference]{IEEEtran}
\IEEEoverridecommandlockouts
\usepackage[T1]{fontenc}
\usepackage{datetime}
\usepackage{listings}
\usepackage{graphicx}
\usepackage{float}
\usepackage{subfigure} 
\usepackage{cite}
\usepackage{caption} 
\usepackage{color}
\usepackage{amsthm}
\usepackage{bm}
\usepackage{mathrsfs}
\newtheorem{theorem}{Theorem}

\newtheorem{lemma}[theorem]{Lemma}
\usepackage[noend]{algpseudocode}

\usepackage{algorithmicx,algorithm}

\usepackage{amsmath}

\usepackage[cmintegrals]{newtxmath}  

\def\BibTeX{{\rm B\kern-.05em{\sc i\kern-.025em b}\kern-.08em
    T\kern-.1667em\lower.7ex\hbox{E}\kern-.125emX}}
\begin{document}
\title{Sparse Signal Processing for Massive
	Connectivity via Mixed-Integer
	Programming} 

\author{Shuang Liang, Yuanming Shi, and Yong Zhou\\
	School of Information Science and Technology, ShanghaiTech University, Shanghai \\
	\{liangshuang, shiym, zhouyong\}@shanghaitech.edu.cn
}

\maketitle

\begin{abstract}
Massive connectivity is a critical challenge of Internet of Things (IoT) networks. In this paper, we consider the grant-free uplink transmission of an IoT network with a multi-antenna base station (BS) and a large number of single-antenna IoT devices. Due to the sporadic nature of IoT devices, we formulate the joint activity detection and channel estimation (JADCE) problem as a group-sparse matrix estimation problem. Although many algorithms have been proposed to solve the JADCE problem, most of them are developed based on compressive sensing technique, yielding suboptimal solutions. In this paper, we first develop an efficient weighted $l_1$-norm minimization algorithm to better approximate the group sparsity than the existing mixed $l_1/l_2$-norm minimization. 
Although an enhanced estimation performance in terms of the mean squared error (MSE) can be achieved, the weighted $l_1$-norm minimization algorithm is still a convex relaxation of the original group-sparse matrix estimation problem, yielding a suboptimal solution. 
To this end, we further reformulate the JADCE problem as a mixed integer programming (MIP) problem, which can be solved by using the branch-and-bound method. 
As a result, we are able to obtain an optimal solution of the JADCE problem, which can be adopted as an upper bound to evaluate the effectiveness of the existing algorithms. Moreover, we also derive the minimum pilot sequence length required to fully recover the estimated matrix in the noiseless scenario. Simulation results show the performance gains of the proposed optimal algorithm over the proposed weighted $l_1$-norm algorithm and the conventional mixed $l_1/l_2$-norm algorithm.  Results also show that the proposed algorithms require a short pilot sequence than the conventional algorithm to achieve the same estimation performance. 
\end{abstract}

\section{Introduction}
 Providing massive connectivity for machine-type communications (MTC) is a key issue that needs to be addressed in the fifth-generation (5G) wireless networks\cite{massive}. The typical applications of the massive MTC use cases include smart homes, wearables, environment sensing and healthcare. With the increasing popularity of Internet of Things (IoT) services and the continuous reduction of the cost of IoT devices, the number of IoT devices is expected to reach 75.4 billion by 2025\cite{smart}. However, due to the limited number of orthogonal signature sequences, it is impossible to assign mutually orthogonal signature sequences to all devices\cite{grantfree}. A key characteristic of IoT data traffic is the sporadic transmission\cite{sparse}, i.e., only a portion of IoT devices are active at any instant. Exploiting this feature, a sparse linear model can be established for massive connectivity, which enables the joint device activity detection and channel estimation.
\par Grant-based random access was initially proposed to coordinate multiple IoT devices to access the network. However, this scheme cannot be applied to support the massive connectivity due to the following reasons. First, the excessive signaling overhead leads to large communication overhead. Second, it is generally infeasible to assign orthogonal signature sequences to all IoT devices. To solve these problems, the 3rd generation partnership project (3GPP) for 5G new radio (NR) proposed the grant-free random access scheme \cite{grantfree}. With grant-free random access, various methods\cite{shi2020low} have been proposed to tackle the device activity detection and channel estimation. Particularly, the authors in \cite{5671560} adopted the compressed sensing (CS) technique to achieve joint device activity detection and channel estimation. To improve the efficiency, the authors in \cite{bayesian} proposed a modified Bayesian compressed sensing algorithm. In order to improve the computational efficiency under massive connectivity, the approximate message passing (AMP) algorithm was proposed to solve the CS problems\cite{amp1,amp2,amp3}. However, AMP algorithms often fail to converge when the preamble signature matrix is mildly ill-conditioned or non-Gaussian\cite{fletcher2018plug}. The authors in \cite{lasso1,lasso2} introduced a mixed $l_1/l_2$-norm convex relaxation method to reformulate the problem as a form of group LASSO. Many methods can be used to solve the LASSO problem, such as interior-point\cite{lasso1} method and iterative shrinkage thresholding
algorithm (ISTA)\cite{ista} method. However, these methods normally take a large number
of iterations for convergence. In order to accelerate the convergence, the deep learning based methods have recently been proposed for sparse signal processing. The data-driven method such as learned iterative shrinkage thresholding algorithm (LISTA)\cite{lista} was proposed to speed-up the sparse signal recovery. Moreover, \cite{chen2018theoretical} established the weight coupling structure for LISTA and proved that LISTA achieves a linear convergence rate. The authors in\cite{lista2} developed LISTA for group sparsity estimation problem. However, most of the existing methods relied on suboptimal formulations, where the optimal solution cannot be achieved.

\par In this paper, we consider the grant-free uplink transmission of a single-cell IoT network. Taking into account the sporadic transmission, we formulate a group-sparsity estimation problem for joint device activity detection and channel estimation, and propose a novel method to obtain the optimal solution. We reformulate the original problems based on $l_0$-norm-based. Motivated by \cite{exact}, we develop the exact sparse approximation problems via Mixed-Integer Programs (MIP) for group sparse problem, and use Branch and Bound ($B\&B$) algorithm solve it. In addition, we derive the minimum pilot sequence length required to fully recover the estimated matrix in the noiseless scenario. Simulation results demonstrate that our method outperforms classical methods and the exact estimation can be guaranteed by using shorter sequence length. 
\section{System Model And Problem Formulation}
\subsection{System Model}
Consider the grant-free uplink transmission of a single-cell IoT network consisting of one $M$-antenna BS and $N$ single-antenna IoT devices. We denote $\mathcal{N}=\{1,\dots,N\}$ as the index set of IoT devices. We consider full frequency reuse and quasi-static fading channels. In each transmission block, the transmission activity of each device is independent. We denote $a_n$ as the indicator of the activity of device $n$, where $a_n=1$ if the device $n$ is active and $a_n=0$ otherwise. Assume the transmissions of active devices are synchronized, we express the pilot signals superimposed at the BS as 
\begin{gather}
\bm{y}(l)=\sum_{n=1}^{N}a_n \bm{h}_n s_n(l)+\bm{z}(l),l=1,
\dots,L
\end{gather}
where $\bm{y}(l)\in \mathbb{C}^M$ denotes the $l$-th symbol received by the BS, $\bm{h}_n\in \mathbb{C}^M$ denotes the channel coefficient vector of the link between IoT device $n$ and the BS, $s_n(l)\in \mathbb{C}$ denotes the $l$-th signature symbol transmitted by device $n$, $L$ denotes the length of the signature sequence, and $\bm{z}(l)\in \mathbb{C}^M$ denotes the additive white Gaussian noise (AWGN) vector at the BS.
\par Since the number of devices is usually much larger than the length of the signature sequence, i.e., $N\gg L$, allocating mutually orthogonal signature sequences
to all IoT devices is not possible. Hence, we assume that each device
is assigned a unique signature sequence. We generate the signature
sequences according to the independent and identically distributed (i.i.d.) complex Gaussian distribution, i.e., $s_n(l)\sim \mathcal{CN}(0,1)$. These signature sequences are generally non-orthogonal.
\par By accumulating the signal vectors over $L$ time slots, we denote the aggregated received pilot signal matrix
$\bm{Y}=[\bm{y}(1),\dots,\bm{y}(L)]^T\in \mathbb{C}^{L\times M}$, the channel matrix 
$\bm{H}=[\bm{h}_1,\dots,\bm{h}_N]^T\in \mathbb{C}^{N\times M}$, the additive noise matrix 
$\bm{Z}=[\bm{z}(1),\dots,\bm{z}(L)]^T\in \mathbb{C}^{L\times M}$, and the pilot matrix
$\bm{S}=[\bm{s}(1),\dots,\bm{s}(L)]^T\in \mathbb{C}^{L\times N}$, where $\bm{s}(l)=[s_1(l),\dots,s_N(l)]^T\in \mathbb{C}^N$.
Thus, (1) can be rewritten as 
\begin{gather}
\bm{Y}=\bm{SX}+\bm{Z},
\end{gather}
where matrix $\bm{X}=\bm{AH} \in C^{N\times M}$ with 
$\bm{A}=\mathrm{diag} \left( a_1,\dots,a_N\right)\in \mathbb{R}^{N\times N}$ being the diagonal activity matrix. Hence, matrix $\bm{X}$ endows with a group sparse structure. In other words, each column of matrix $\bm{X}$ is sparse and all columns have the same sparse structure. The goal is to detect the device activity $\bm{A}$ and estimate the channel matrix $\bm{H}$ by recovering $\bm{X}$ from the noisy observation $\bm{Y}$. 

\subsection{Problem Formulation}
Because of the multi-antenna BS, we cannot directly induce the sparsity of matrix $\bm{X}$ by using the widely adopted $l_0$-norm. Instead, by exploiting the group sparse structure, we use the $l_0$-norm of first column of matrix ${\bm{X}}$ to represent the sparsity of matrix ${\bm{X}}$. Taking into account the structure of matrix ${\bm{X}}$, the optimization problem can be written as 
\begin{eqnarray}
\mathscr{P}_1:
\mathop {\textrm{minimize}}_{{{\bm{X}}}}&& ||{\bm{X}}^1||_0\nonumber\\
\textrm{subject to}&&||{\bm{Y}}-{\bm{S}}{\bm{X}}||_F\leq \epsilon\nonumber,\\
&& I({\bm{X}}^j)=I({\bm{X}}^k), {\forall j\ne k},
\end{eqnarray}
where ${\bm{X}}^j$ denotes the $j$-th column of ${\bm{X}}$, $\epsilon$ is a predefined error limit, and $I(\bm{x})$ is the support (non-zero indicator variable) of vector $\bm{x}$. Such a problem is NP-hard due to the inherent combinatorial nature. The existing methods\cite{massive} replace the $l_0$-norm objective function by its $l_1$-norm, and approximate problem $\mathscr{P}_1$ as follows
\begin{eqnarray}
\mathscr{P}_2:
\mathop {\textrm{minimize}}_{{\bm{X}}}&& \Phi({\textbf{X}}):=\sum_{i=1}^{N}||{\bm{X}}_i||_2\nonumber\\
\textrm{subject to}&&||{\bm{Y}}-{\bm{S}}{\bm{X}}||_F\leq \epsilon,
\end{eqnarray}
where ${\bm{X}}_i$ is the $i$-th row of ${\bm{X}}$ and $\Phi({\bm{X}})$ induces the group sparsity via the mixed $l_1/l_2$-norm. The $l_2$-norm $||{\bm{X}}_i||_2$ bounds the magnitude of ${\bm{X}}_i$ and the $l_1$-norm induces the sparsity of 
\begin{eqnarray}
[||{\bm{X}}_1||_2,...,||{\bm{X}}_N||_2].
\end{eqnarray}
Problem $\mathscr{P}_2$ is a convex problem, which can be solved by many existing algorithms. With the estimated matrix ${\bm{X}}$, the device activity can be obtained by setting $a_i=1$ if $||{\bm{X}}_i||_2\geq \gamma_0$ for a predefined threshold $\gamma_0$, and $a_i=0$ otherwise. 
\par To reduce the number of measurements required for accurate recovery, the authors in \cite{reweighted} formulated a reweighted $l_1$-norm-minimization problem, which can be reformulated as a group sparse optimization problem.  
\begin{eqnarray}
\mathscr{P}_{3}:
\mathop {\textrm{minimize}}_{{\bm{X}}}&&\sum_{i=1}^{N} w_i||{\bm{X}}_i||_2\nonumber\\
\textrm{subject to}&&||{\bm{Y}}-{\bm{S}}{\bm{X}}||_F\leq \epsilon,
\end{eqnarray}
where $w_i$ is the weight of $||{\bm{X}}_i||_2$. After each iteration, we update $w_i=\frac{1}{||{\bm{X}}_i||_2}$ to get a new problem and repeat the process until convergence. We summarize the algorithm for solving problem $\mathscr{P}_{3}$  in Algorithm 1.
\begin{algorithm}[t]
	\caption{Reweighted $l_1/l_2$ minimization problem } 
	\hspace*{0.02in} {\bf Input:} 
	 parameters ${\bm{Y}}$, ${\bm{S}}$
	\begin{algorithmic}[1]
		\State Set the iteration count $n$ to zero and $w_i^{(0)}=1,i=1,\dots,N$ 
		\State Solve the weighted $l_1/l_2$ minimization problem $\mathscr{P}_{3}$
		\State Update the weights: $w_i^{(n)}=\frac{1}{||{\bm{X}}_i||_2}$, $i=1,\dots,N$
		\State Terminate until convergence or when $n$ attains a specific number. Otherwise, update iteration count $n$ and go to step 2.
	\end{algorithmic}
	\hspace*{0.02in} {\bf Output:} 
estimation of ${\bm{X}}$
\end{algorithm}
\par It is obvious that problems $\mathscr{P}_2$ and $\mathscr{P}_3$ are convex relaxation of problem $\mathscr{P}_1$  and their solutions are suboptimal to the original problem. To find the global optimal solution, which can be adopted as an upper bound to evaluate the effectiveness of the existing algorithms, we focus on solving the exact optimization of the $l_0$-norm-based problem through Mixed-Integer Quadratically Constrained Program (MIQCP)\cite{exact}.
\section{Proposed Global Optimal Algorithm}
In this section, we reformulate  problem  $\mathscr{P}_1$  as an MIQCP problem\cite{exact} and solve the problem by using the branch and bound algorithm. 

\subsection{Reformulation of $l_0$-norm-based Objective Function}
In this subsection, we reformulate the objective function based on $l_0$-norm to the summation of binary variables. According to\cite{lista2}, we can equivalently convert a complex variable to a real variable. To facilitate the reformulation, we only consider the real variables.

We first present the boundedness assumption and the `Big-$\beta$' that enable us to express $\mathscr{P}_1$ as a MIQCP problem. In particular, we introduce an additional binary optimization variable $\bm{b}\in\{0,1\}^N$, such that 
\begin{eqnarray}
b_i=0 \iff {\bm{X}}_i=\bm{0}.
\end{eqnarray}
Hence, the non-linear sparsity measure of ${\bm{X}}$ can be  equally represented by the linear term $\sum_{i=1}^N b_i$. The minimization of $||{\bm{X}}^{1}||_0$ should be transformed into (in)equality formulation to be compatible with MIP. One standard way is to assume that each entry of  ${\bm{X}}$ satisfies the following constraints for some sufficiently large pre-defined value $\beta>0$:
\begin{eqnarray}
-\beta b_i \leq {\textbf{X}}_{ij} \leq \beta b_i.
\end{eqnarray}
This assumption ensures that problem  $\mathscr{P}_1$ admits bounded optimal solutions. Parameter $\beta$ needs be large enough in order that
$\max\{|{\bm{X}}_{ij}|\} \leq \beta$ at any desirable optimal solution. At the same time, to improve the computational efficiency, the bound must be as tight as possible. Thus, it is critical to tune the value of $\beta$. In the problem addressed in this paper, satisfactory results can be obtained by using a simple empirical
3-$\sigma$ rule.
The reformulation of $l_0$-norm-based constraints and objective function are obtained through the following two lemmas.
\begin{lemma}
	If the group sparsity of ${\bm{X}}$ is less than or equal to $K$, then there exits a vector $\bm{b}$ that  satisfies (9) and $\sum_{i=1}^{N}b_i \leq K$, and vice versa. Mathematically, we have
\end{lemma}
 \begin{eqnarray}
||{\bm{X}}^{1}||_0\leq K \iff\left\{
\begin{aligned}
&\exists \bm{b}\in\{0,1\}^N \text{such that} \\
&\sum_{i=1}^{N}b_i \leq K,(i) \\
&-\beta b_i \leq {\textbf{X}}_{ij} \leq \beta b_i.(ii)
\end{aligned}
\right.
\end{eqnarray}
\begin{proof}
 First, we proof the sufficiency. According to (9), we can always construct vector that $\bm{b}$ satisfies $\sum_{i=1}^{N}b_i \leq K$. Second, we proof the necessity. Let $\bm{b}$ satisfy ($i$) and ($ii$), and suppose $||{\bm{X}}^{1}||_0 > K$. From ($ii$), we have ($b_i=0$) $\Rightarrow$ (${\bm{X}}^i=0$) and (${\bm{X}}^i\not =0$) $\Rightarrow$ ($b_i\not = 0$). Hence $\sum_{i=1}^N b_i > K$ , which contradicts with ($i$). Consequently, we have $||{\bm{X}}^1||_0 \leq K$.
\end{proof}
\begin{lemma}
	 Minimizing the group sparsity of ${\bm{X}}$ is equivalent to minimizing the sum of binary variables $b_i$, i.e.,
\end{lemma}
\begin{eqnarray}
\mathop{\textrm{minimize}}_{{\bm{X\in F}}}||{\bm{X}}^1||_0\iff \left\{
\begin{aligned}
&\mathop{\textrm{minimize}}_{\bm{X\in F}}\sum_{i=1}^N b_i\\
&s.t. -\beta b_i \leq {\textbf{X}}_{ij} \leq \beta b_i \\
&\bm{b}\in\{0,1\}^N
\end{aligned}
\right.
\end{eqnarray}
where $F$ represents the feasible domain of the problem under consideration. 
\begin{proof}
 First, we proof the sufficiency. Suppose $\mathop{\textrm{argmin}}||{\bm{X}}^1||_0$=$\hat{\bm{X}}$, and $||\hat{\bm{X}}^1||_0=K_0$. From Lemma 1, one has $\sum_{i=1}^Nb_i^*$=$K_0$, if there is $\sum_{i=1}^N b_i<K_0$, then there is $||{\bm{X}}^1||_0<K_0$, which contradicts $\mathop{\textrm{argmin}}||{\bm{X}}^1||_0$=$\hat{\bm{X}}$. Consequently, $\mathop{\textrm{argmin}}{\sum_{i=1}^Nb_i}$=${\bm{b^*}}$. Second, the proof of necessity is similar to that of sufficiency. 
\end{proof} 
\par Based on Lemma 1 and Lemma 2, we turn the objective function of minimizing sparsity into the form that is compatible with MIP.
\subsection{Reformulation of $l_F$-norm-based Constraints}
In this subsection, we reformulate the constraint of (4) to the standard form of MIQCP.
\par In the multi-antenna scenario, $||{\bm{Y}}-{\bm{S}}{\bm{X}}||_2\leq \epsilon$ cannot be directly expanded as in the single-antenna scenario. By constraining the error $||{\bm{Y}}-{\bm{S}}{\bm{X}}||_2$, we have
\begin{eqnarray}
&({\bm{X}}^i)^T{\bm{S}}^T{\bm{S}}{\bm{X}}^i-2({\bm{Y}}^i)^T{\bm{S}}{\bm{X}}^i \leq \epsilon_i^2-||{\bm{Y}}^i||_2^2,\\
&\sum_{i=1}^{M}\epsilon_i^2=\epsilon^2.
\end{eqnarray}
\par With the above transformation, problem $\mathscr{P}_1$ is equal to MIQCP $\mathscr{P}_4$ as follows
\begin{eqnarray}
\mathscr{P}_4:\nonumber
\mathop {\textrm{minimize}}_{\bm{X},{\bm{b}}}&& \bm{1}^T\bm{b}\nonumber\\
\textrm{subject to}&&({\bm{X}}^i)^T{\bm{S}}^T{\bm{S}}{\bm{X}}^i-2({\bm{Y}}^i)^T{\bm{S}}{\bm{X}}^i \leq \epsilon_i^2-||{\bm{Y}}^i||_2^2,\nonumber\\
&& \sum_{i=1}^{M}\epsilon_i^2=\epsilon^2,\nonumber\\
&& -\beta b_i \leq {\textbf{X}}_{ij} \leq \beta b_i,\nonumber\\
&& b_i\in \{0,1\}.
\end{eqnarray}
\par The MIP problems are NP-hard. Transforming problem $\mathscr{P}_1$ to problem $\mathscr{P}_4$ does not reduce the  computational complexity. However, such a reformulation can help prove the optimality of the returned solutions. Due to the progress in both hardware and algorithm, MIP problems can be solved by many state-of-the-art solvers such as CPLEX\cite{bixby2012brief}, SCIP, and GUROBI. In this paper, we choose CPLEX because of its solving speed. It is developed by IBM and freely available for research.
\par The main method that the CPLEX MIP solver used is the branch-and-bound ($B\&B$) algorithm. It implements the $B\&B$ strategy by relaxing the integer variables. By fixing some integer variables and relaxing another integer variable, each branch generates a subproblem. Lower bounds are obtained by solving each subproblem. If the returned solutions are just integers, then upper bounds are obtained.
Some branches are cut if their solutions are outside the bound. Repeating the process until the gap is zero or small enough, the optimal solution can be found. In this process, CPLEX incorporates several technologies to reduce the solution time .
\subsection{Special Case}
\par We consider a noiseless scenario in this subsection. By setting $\epsilon=0$, problem $\mathscr{P}_4$ is reduced to a Mixed-Integer Linear Programming (MILP) as follows
\begin{eqnarray}
\mathscr{P}_5:
\mathop {\textrm{minimize}}_{{\bm{X}},{\bm{b}}}&& \bm{1}^T\bm{b}\nonumber\\
\textrm{subject to}&& -{\bm{S}}{\bm{X}}+{\bm{Y}}=0,\nonumber\\
&& -\beta b_i \leq {\textbf{X}}_{ij} \leq \beta b_i\nonumber,\\
&& b_i\in \{0,1\}.
\end{eqnarray}
\begin{lemma}
In the noiseless scenario, we have $\bm{Y}=\bm{S}\bm{X}^*$. Given $\bm{Y}$ and $\bm{S}$, by minimizing $||{\bm{X}}^1||_0$, we can not recover $\bm{X}^*$ if $L \leq K$. In other words, to recover $\bm{X^*}$, the minimum length of the sequence is $L=K+1$. 
\end{lemma}
\begin{proof}
We define that $\bm{S}[:,I_l]$ and $\bm{X}[I_l,:]$ are $l$ columns of $\bm{S}$ and $\bm{X}$ respectively, and $I_l$ is the corresponding index set. By defining the non-zero row index set of $\bm{X^*}$ as $I_K$, we get $\bm{Y}=\bm{S}[:,I_K]\bm{X^*}[I_K,:]$, $\bm{X^*}[I_K^\mathsf{c},:]=0$. Solving the problem $\mathscr{P}_5$ is equivalent to finding the minimum of $l$ to ensure $I_l$ satisfies $\bm{Y}=\bm{S}[:,I_l]\bm{X}[I_l,:]$ and  $\bm{X}[I_l^\mathsf{c},:]=0$.
\par The signature sequence is generated according to the complex Gaussian distribution, i.e, $\bm{S}\sim \mathcal{CN}(0,1)$. We define $P(\bm{S}^i=c\bm{S}^j)$ as the probability of $\bm{S}^i=c\bm{S}^j$ ($i\not=j$ and $c$ is a constant). Because of the continuity of the probability density function, we get $P(\bm{S}^i=c\bm{S}^j)=0$. Similarly, we have $P(\bm{S}^i=c_1\bm{S}^j+c_2\bm{S}^k)=0$. In other words, any $l$ ($l<L$) columns of $\bm{S}$ are linearly independent. 
\par In this part, we talk about the solution of the overdetermined equations $\bm{Y}=\bm{S}[:,I_l]\bm{X}[I_l,:]$ ($l<L$). Each column of $\bm{Y}$ is a linear combinations of $I_K$ columns in $\bm{S}$. We have $R([\bm{S}[:,I_l]|Y]$)=$R([\bm{S}[:,I_l]|\bm{S}[:,I_K]]$), where $R(\cdot)$ is the range space of the input matrix and $[\bm{S}[:,I_l]|Y]$ is the augmented matrix. Since any columns of $\bm{S}[:,I_l]$ are linearly independent, we have $R([\bm{S}[:,I_l]|Y]$)=$R(\bm{S}[:,I_l]$) and there is a solution of the overdetermined equations if and only if $I_K\subset I_l$.
\par Consider the situation of $L\le K$ and $L\textgreater K$ respectively. In the situation of $l \textless L \textless K$, $I_K \not \subset I_l$ and there is no solution of the overdetermined equations $\bm{Y}=\bm{S}[:,I_l]\bm{X}[I_l,:]$. When $l=L\le K$, the equations $\bm{Y}=\bm{S}[:,I_l]\bm{X}[I_l,:]$ have a solution $\bm{X}[I_l,:]=\bm{S}[:,I_L]^{-1}Y$. Due to the randomness of $I_L$, there are $C_N^L$ solutions and the sparsity of the solutions is less than or same as $\bm{X}^*$. Thus, $\bm{X}^*$ can not be recover exactly in this case. Consider the situation of $L\textgreater K$. When $k<L$, the equations $\bm{Y}=\bm{S}[:,I_k]\bm{X}[I_k,:]$ have a solution $\bm{X}^*[I_K,:]$ if and only if $I_{k}$=$I_{K}$. In this case, we can recover $\bm{X^*}$. 
\end{proof}
\section{Simulation Results}
In this section, we present the simulation results of the proposed global optimal algorithm for solving problem $\mathscr{P}_4$, and compare the results obtained by using the convex relax method for solving problem $\mathscr{P}_2$ and $\mathscr{P}_3$. 
In the simulations, the signature sequence is generated according to the complex Gaussian distribution, i.e, $\bm{S}\sim \mathcal{CN}(0,1)$, and the channels suffer from independent Rayleigh fading, i.e., $\bm{H}\sim \mathcal{CN}(0,1)$. We set the  total number of devices $N$, the number of active devices $K$, and  the number of antennas at the BS $M$ to be 30, 5, and 2, respectively. We adopt normalized mean square error (NMSE) to evaluate the performance of those methods:
\begin{gather}
NMSE=10\log_{10}\Bigg(\frac{E||{\bm{X}}-{\bm{X}^*}||_F^2}{E{||{\bm{X}}^*||_F^2}}\Bigg),
\end{gather}
where ${\bm{X}}$ is the estimate solution and ${\bm{X}}^*$ is the ground truth.

\begin{figure}[t]
	\centering
	\includegraphics[scale=0.4]{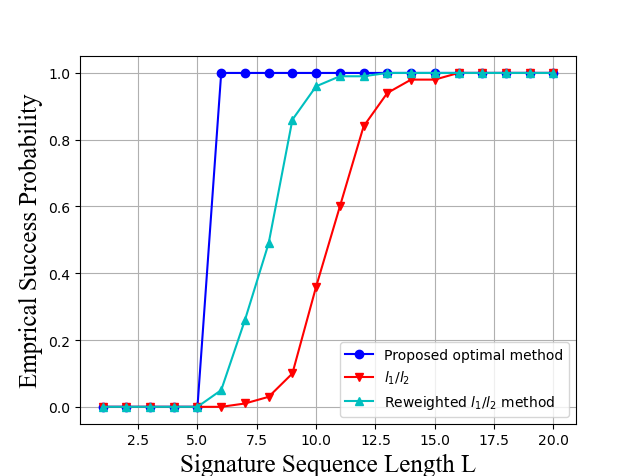}
	\caption{Success probability versus signature sequence length in the noiseless case when $N=30,K=5$ and $M=2$.}
	\label{fig:label}
	\vspace{-5mm}
\end{figure}

\par We show the success probability of different methods versus the length of signature sequence $L$ in a noiseless scenario, as shown in Fig. 1. We declare successful recovery if $||\hat{\textbf{X}}-{\textbf{X}}||_F$ $\leq 10^{-5}$ and each point is obtained by averaging over 100 times. In Fig. 1, we can observe that the recovery success probability of each method increases as the sequence length increases. To recover $\bm{X}^*$ successfully, the shortest sequence lengths of proposed method, $l_1/l_2$ method and reweighted $l_1/l_2$ method are $L=6, L=16$, and $L=11$, respectively. The performance improves $30\%$ compare with $l_1/l_2$ method by using the reweighted method. The proposed method requires the shortest sequence length for successful recovery, which is far less than others.
\begin{figure}[t]
	\centering
	\includegraphics[scale=0.4]{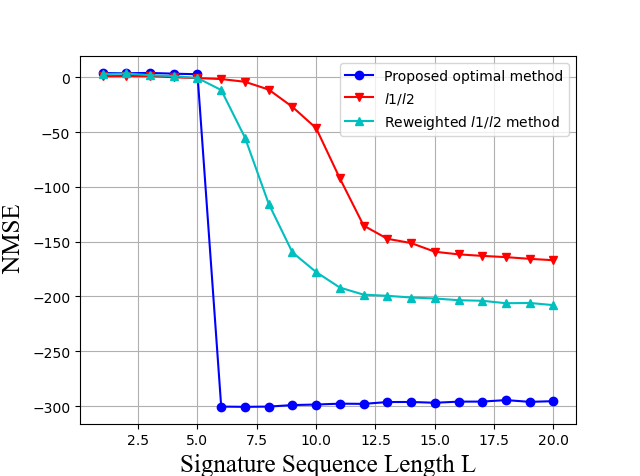}
	\caption{NMSE in the noiseless scenario when $N=30,K=5,M=2$ and $SNR=30$.}
	\label{fig:label}
	\vspace{-5mm}
\end{figure} 
\par To compare the performance more accurately, we show the NMSE of different methods versus the length of signature sequence $L$ in a noiseless scenario, as shown in Fig. 2. 
We can observe that the proposed method not only requires a shorter signature sequence, but also achieves a smaller NMSE. 
 The NMSE of the proposed method has a sharp drop when $L=6$ and it converges to $-300$ directly.
\begin{figure}[t]
	\centering
	\includegraphics[scale=0.4]{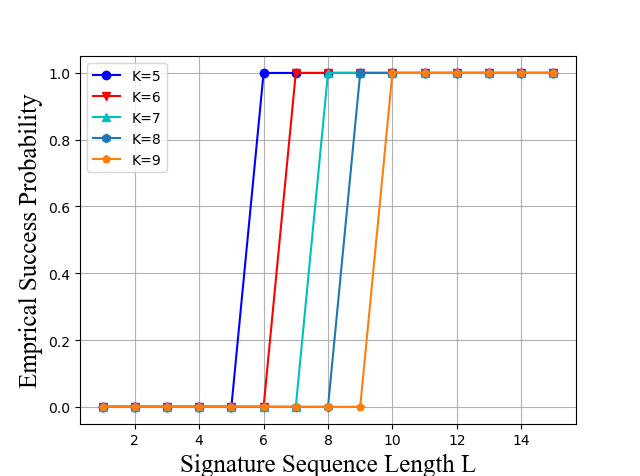}
	\caption{The relationship between $K$ and $L$ in the noiseless scenario when $N=30, M=2$.}
	\label{fig:label}
	\vspace{-5mm}
\end{figure}
\par It is obvious from Fig. 3 that ${\bm{X}}$ can be exactly recovered when $L=K+1$,  which is consistent with Lemma 3. Without any priori knowledge, we need at least $L=K+1$ to achieve exact recovery. We can reduce the overhead effectively for massive connectivity by using the proposed method.
\begin{figure}[t]
	\centering
	\includegraphics[scale=0.4]{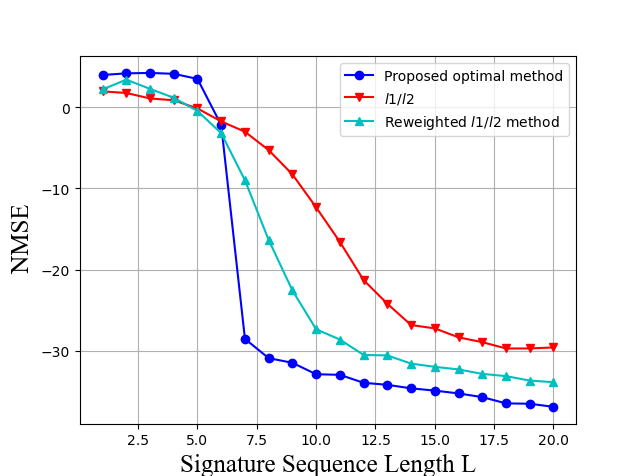}
	\caption{NMSE in the noisy scenario when $N=30,K=5,M=2$ and $SNR=30$.}
	\label{fig:label}
	\vspace{-5mm}
\end{figure}

\par According to Figs. 2 and 4, we observe that the performance is worse than that in the noiseless scenario. But the performance of our method is also the best in the noisy scenario. It also has a requires a shorter signature sequence and a smaller NMSE. Using the proposed method, we can obtain an optimal solution of problem $\mathscr{P}_1$, and we can adopt it as an upper bound to evaluate the effectiveness of the existing algorithms.
\section{Conclusion}
In this paper, we proposed a novel method to address the group sparse estimation problem in IoT networks. We transform the  original problem into an equivalent MIP problem and use CPLEX slover to solve it. Simulation results show that the performance of our method is better than $l_1$-reweighted method and convex relaxation method in both noiseless and noisy scenarios. It provides an optimal solution about the number of measurements to sparse estimation problem without priori. We can adopt the solution as an upper bound to evaluate the effectiveness of the existing algorithms. In noiseless scenario, we just need sequence length larger than the number of active devices to recover the estimated matrix exactly.
\appendices

\bibliographystyle{IEEEtran}
\bibliography{ref}
\end{document}